\documentclass[twocolumn,pra,aps,amssymb]{revtex4}
\usepackage{amsmath,amssymb,amsfonts,amsthm}
\usepackage{multirow}
\usepackage{verbatim}
\usepackage{url}
\usepackage{comment}
\usepackage{hyperref}
\usepackage{xcolor}
\begin{document}
% Dirac Kets
\newcommand{\ket}[1]{\ensuremath{\left|#1\right\rangle}}
\newcommand{\bra}[1]{\ensuremath{\left\langle#1\right|}}
\newtheorem{definition}{Definition} 
\newtheorem{theorem}{Theorem} 

%\begin{comment}
\title{Proposal for Quantum Rational Secret Sharing}
\author{Arpita Maitra$^1$, Sourya Joyee De$^2$, Goutam Paul$^2$ and Asim K. Pal$^1$}
\affiliation{$^1$Management Information Systems Group,\\
Indian Institute of Management Calcutta, India.\\
Email: \{arpitam, asim\}@iimcal.ac.in\\
$^2$ Cryptology \& Security Research Unit,\\
R. C. Bose Centre for Cryptology \& Security,\\
Indian Statistical Institute, Kolkata,\\
Email: sjoyeede@gmail.com, goutam.paul@isical.ac.in
}
%\end{comment}

%\title{}
%\author{}

\begin{abstract}
A rational secret sharing scheme is a game in which each party responsible for reconstructing a secret tries to maximize his utility by obtaining the secret alone. Quantum secret sharing schemes, either derived from quantum teleportation or from quantum error correcting code, do not succeed when we assume rational participants. This is because all existing quantum secret sharing schemes consider that the secret is reconstructed by a party chosen by the dealer. In this paper, for the first time, we propose a quantum secret sharing scheme which is resistant to rational parties. The proposed scheme is fair (everyone gets the secret), correct and achieves strict Nash equilibrium.
\end{abstract}
\maketitle

\section{Introduction and Motivation}

Secret sharing is an important primitive in cryptography. It can be considered as a special case of secure multiparty computation~\cite{Gordon,AL11,GroceK} 
which has applications in electronic voting, cloud computing, online auction etc. Recently, significant effort has been given towards bridging the gap between two apparently unrelated domains, namely, cryptography and game theory~\cite{AL11,GroceK,Katz12}. Cryptography deals with the `worst case' scenario making the protocols secure against malicious behavior of a party. However, in game theoretic perspective, a protocol is designed against the rational deviation of a party.

In rational domain there is no concept of trust.
Rational players are classified as neither `good' nor `bad'. They participate in the game with a motivation to maximize their utility. In cryptography, one may consider this as a special type of attack vector. However, this does not impose any special condition on adversary, it rather adds more flexibility to the adversary.

In~\cite{Grabbe} it was commented that quantum secret sharing can be treated as a game between the legitimate parties. Very recently, Brunner and Linden~\cite{BL} showed a deep link between quantum physics and game theory.  By bringing quantum mechanics into the game, they showed that players who can use quantum resources, such as entangled quantum particles, can outperform classical players. This is because of the fact that, in classical domain, the security depends on some computational hardness and thus is conditional. On the other hand, in quantum domain, the security comes from the laws of physics and thus is unconditional. In this paper, for the first time, we introduce the rationality concept of game theory in quantum secret sharing.

A $(t,n)$ or $t$-out-of-$n$ threshold secret sharing scheme~\cite{Shamir,Blakely} comprises the 
distribution of shares of a secret $s$ among $n$ players $P_1, \ldots, P_n$, 
such that at least $t$ of these players must communicate their shares to each
other to reconstruct the secret. An example of such a secret sharing scheme is Shamir's scheme~\cite{Shamir} that uses the concept of polynomial interpolation for generation and 
distribution of shares of the secret by a dealer and subsequent 
reconstruction of the secret by the players. Players that are `good' or `honest'
cooperate to reconstruct the secret, while players that are `bad' or malicious 
do not cooperate \cite{HT}. So, for successful reconstruction of the secret, at most ($n-t$) players can be `bad'.

In classical threshold secret sharing, Halpern and Teague~\cite{HT} 
introduced the concept of rational players.
Each rational party wishes to learn the secret while allowing as few others as possible to learn the secret. Halpern and Teague~\cite{HT} showed that in the presence of rational players, Shamir's scheme fails. 
Specifically, no rational player has the incentive to send his share during secret reconstruction. From the viewpoint of each player $P_i$, either ($t-1$) other players send their shares or they do not. If they send, then $P_i$, even without sending his own share, can reconstruct the secret for himself without allowing these ($t-1$) players to reconstruct. If they do not send, then none of the players can reconstruct the secret. So from each player $P_i$'s point of view, not sending his share weakly dominates sending his share. Thus the Nash equilibrium
achieved in Shamir's secret sharing corresponds to the case when nobody sends
anything to each other. To mitigate this problem, the
authors of~\cite{HT} introduced the concept of~\textit{rational secret sharing} (RSS). Its application in secure multiparty computation is known as \textit{rational multi-party computation} or RMPC and has been an active area of research~\cite{ADGH,AL,GK,KN1,LT,FKN,OPRV,LS} in recent times. 

The idea of quantum secret sharing (QSS) of a single qubit was first due to 
Hillery et al.~\cite{HBB99} using three and four qubit GHZ states. Later, this 
process was investigated by Karlsson et al.~\cite{KK99} using three particle 
entanglement, Cleve et al.~\cite{C99} using a process similar to error 
correction and Zheng~\cite{Z06} using W state. The QSS of an arbitrary 
two-qubit state was proposed by Deng et al.~\cite{D05}
using two GHZ states.
QSS using cluster states was demonstrated by Nie~\cite{NS11},  
Panigrahi~\cite{PM11,PS11} and Han~\cite{H12}. Recently, two qubit QSS 
was discussed using arbitrary pure or mixed resource states~\cite{ZL11} and 
asymmetric multipartite state~\cite{Z12}. Note that in $t$-out-of-$n$
QSS, the dealer chooses to reveal the secret to a specific subset of
$t$ parties and not to any arbitrary subset of $t$ parties.

\subsection{QSS with Rational Adversaries}
In QSS, all the parties are `good' or `honest' as they have agreed to reconstruct the secret to the party (or parties) chosen by the dealer. However, if we impose rational behavior of the participants in QSS, it is quite natural for the last player, who generates the secret, to quit with the secret alone. Hence, 
the other players always prefer not to give their shares (either classical bits or quantum bits) and hence the traditional QSS scheme fails if the players 
behave rationally. Like classical case, one may consider this as a special type of attack vector in quantum secret sharing. In the context of quantum secret sharing, it is an important attack vector to consider. However, this does not impose any special condition on adversary, it rather empowers the adversary with
more flexibility. In this paper, for the first time, we propose a quantum secret sharing scheme that resists this kind of attack vector and forces the participants to send the shares, though they are rational in nature. We call this scheme a {\em quantum rational secret sharing} (QRSS) scheme.

In classical domain, the adversary that controls a player may be computationally bounded, but in quantum domain the adversary is always assumed to have unbounded computational power. Because of this, we assume a computationally unbounded adversary throughout the entire paper and modify the security notions in this direction.

\subsection{Security Issues}
In classical RSS protocols, two type of settings are considered. One is called {\em fail-stop} setting and the other is known as 
{\em Byzantine setting}. In fail-stop setting, a player may abort early in an attempt to obtain the secret alone but does not send false shares of the secret.  Whereas in Byzantine setting, a player can behave arbitrarily, i.e., he can abort early or can fabricate a false share.

For share generation, rational multiparty computation exploits the idea of 
Shamir's secret sharing, the security of which comes from the 
{\em interpolation theorem}~\cite{Shamir,TW88}. Thus, it does not depend on 
some unproven hypothesis on computational hardness. 
However, in~\cite{TW88} Tompa and Woll showed that in Shamir's scheme, any $(t-1)$ participants can fabricate false shares in the motivation to deceive the 
$t$-th participant to believe in a legal but incorrect secret. In other words, Shamir's basic scheme is not secure against Byzantine players.

One straightforward solution to this problem is to send signed  shares by the distributor (dealer) to the participants. Another approach is to use verifiable
secret sharing~\cite{Micali85}. Note that both these approaches
are based upon unproven assumptions such as the intractability of
integer factorization or the existence of secure encryption schemes.

Interestingly, Tompa and Woll proposed a scheme~\cite{TW88} 
that mitigates the hidden problem of cheating in Shamir's secret sharing without
any unproven assumption. They showed that the probability of undetected 
cheating can be made less than $\epsilon$, for any $\epsilon > 0$, 
by suitably choosing a large prime (that depends on $\epsilon$) 
as the modulus of the underlying field.

Thus, the security issues of classical RSS can be sub-divided into the following notions: 1) security of the underlying secret secret sharing, 2) security of a signature scheme, 3) security against rational players. The security issues of $1$ and $2$ have been discussed in a number of literatures~\cite{Shamir,TW88,Micali84}. This is why the works on rational secret sharing~\cite{HT,ADGH,AL,GK,KN1,LT,FKN,OPRV,LS} have concentrated only on the security of the rational 
part, that is formalized in terms of fairness, correctness and Nash equilibrium~\cite{AL}. In this paper, we follow the same approach in the quantum domain. 

In quantum domain,  we impose rationality issues on the top of the quantum secret sharing model which exploits quantum error correcting code. Thus, the security of the secret sharing part comes from the security of the quantum error correcting code, specifically CSS code~\cite{C99,S09,SP00}. Note that as we exploit the quantum error correcting code to encode the secret, no unauthorized party can extract any information by subverting one or more authorized parties~\cite{C99,S09}.

Further, the existing RSS literatures~\cite{ADGH,AL,FKN,KN1,LT,LS} deal with various flavours of Nash equilibrium. As we assume computationally unbounded adversary in the quantum domain, we consider strict Nash equilibrium here.

In classical RSS, the dealer signs each share so that no player can give out wrong shares to others. However, in the quantum setting the scenario is different. Typically, quantum signature schemes consider signing either classical
messages~\cite{gott01} or quantum message string with independent 
qubits~\cite{lu05}. In these works, there is no concept of entanglement among the distributed shares, whereas in our proposed scheme, the shares are entangled. It is not yet known how to sign such qubits which contain the information of the secret, as any type of measurement on that qubits will destroy the entanglement
and hence the information related to the secret. For this reason, we assume that a rational player in the quantum setting is fail-stop by nature, i.e., he may abort early towards the motivation to get the secret alone, but does not send false shares of the secret.

In quantum domain, it is very natural for a player to measure his share as soon as he gets it. However, in this work, we encode the secret by CSS code which takes care of arbitrary error. Thus measuring his qubit in an arbitrary basis gives no advantage to the player. Even with unconditional power of computation, the quantum adversary extracts no information about the secret. Moreover, if he measures the share, he will lose the information stored in the qubit. Thus no player has any incentive to measure his qubit(s) and each player communicates each
share as it is received from the dealer.

\section{Preliminaries}
In this section, we briefly describe classical rational secret sharing and discuss the concepts of rationality, fairness, correctness and equilibrium used in this work. We also extend these concepts in the quantum domain.
 
The dealer in a classical rational secret sharing (RSS) protocol is honest and can be online or offline. An online dealer remains available throughout the secret reconstruction protocol, whereas an offline dealer becomes unavailable after distributing the shares of the secret. Note that an online dealer is not very practical as he repeatedly interacts with the players and such a dealer can directly provide the secret to the players. In 2008, Kol and Naor~\cite{KN1} discussed rational secret sharing in the non-simultaneous channel model and in the presence of an offline dealer, in an information theoretic setting. Almost all the subsequent works~\cite{AL,OPRV,LS,FKN} on rational secret sharing assumed the dealer to be offline.

Rational secret sharing proceeds in two phases: 1) share generation and distribution and 2) secret reconstruction. 

\textit{Share generation and distribution:} If the dealer is online, then at the beginning of each round, he distributes to each player $P_i$ the share of the actual secret with probability $\gamma$ or that of a fake secret with probability $(1-\gamma)$. The value of $\gamma$ is kept secret from the parties and is dependent on the utility values of the parties ~\cite{HT,GK}. An offline dealer distributes to each party $P_i$ a list of shares, one of which is that of the actual secret $s$ and the remaining of fake secrets ~\cite{KN1,FKN,LS}. The position $r$ of this actual share in the lists is not revealed to the players and is chosen according to a geometric distribution $\mathcal{G}(\gamma)$, where the parameter $\gamma$ in turn depends on the utility values of players. The dealer generates shares using Shamir's secret sharing scheme.

\textit{Secret Reconstruction:} In the $j$th round of communication, each player $P_i$ (either simultaneously or non-simultaneously) broadcasts or sends individually to each of the other players (in presence of synchronous, point-to-point channels) the share $s_{ij}$ corresponding to that round. The shares are signed by the dealer. Hence, no player can give out false shares undetected and the only possible action of a player in a round is to either 1) send the message or 2) remain silent. The round in which the shares of the actual secret are revealed and hence the secret is reconstructed is called  revelation or definitive round. When the dealer is offline, players are made aware that they have crossed the revelation round by the reconstruction or exchange of an indicator (a bit in~\cite{KN1}, a signal in~\cite{FKN}). For simultaneous channel model, parties can identify a revelation round as soon as it occurs. However, for non-simultaneous channels, the indication is delayed till the subsequent round to avoid rushing strategy. In this case, the indicator cannot be reconstructed or interpreted by all the players. The player who communicates last during the reconstruction of the indicator is the first and only one to know that the last round was the revelation round. Once he comes to know this, he has no incentive to send his share of the indicator to the other players for reconstruction. Instead, he simply quits. The fact that this player quits signals to the other players that the secret has been reconstructed.

A $(t,n)$ rational secret reconstruction protocol is a pair $(\Gamma, \overrightarrow{\sigma})_{t,n}$, where $\Gamma$ is the game (i.e., specification of allowable actions) and $\overrightarrow{\sigma}$=$(\sigma_1,\ldots,\sigma_n)$ denotes
the strategies followed by the players. We use the notations
$\overrightarrow{\sigma}_{-i} = (\sigma_1,\ldots,\sigma_{i-1},\sigma_{i+1},\ldots,\sigma_n)$ and 
$(\sigma'_i,\overrightarrow{\sigma}_{-i}) = (\sigma_1,\ldots,\sigma_{i-1},\sigma'_i,\sigma_{i+1},\ldots,\sigma_n)$.
The outcome of the game is denoted by $\overrightarrow{o}((\Gamma, \overrightarrow{\sigma})_{t,n})$=$(o_1,\ldots,o_n)$. 
The outcomes of a secret reconstruction game $\Gamma$ with respect to a party $P_i$ are as follows: 1) $P_i$ obtains the secret while others do not; 2) everybody obtains the secret; 3) nobody obtains the secret 4) others obtain the secret while $P_i$ does not and 5) others believe in a fake secret while $P_i$ does not.
The output that no secret is obtained is denoted by $\perp$ and fake secret is denoted by any symbol $\notin \{s,\perp\}$.
 
\subsection{Utilities and Preferences}
The utility function $u_i$ of each party $P_i$ is defined over the set of possible outcomes of the game. The outcomes and corresponding utilities for $t=n=2$ are described in Table~\ref{table: OutcomesRSS}. 
For classical secret sharing, $u_i$ is 
assumed to be polynomial in the security parameter $k$ which is typically 
the size of the secret. Thus, $U_i^{TN}=u_i(1^k, (o_i=s, o_j=\perp))$, $U_i^{TT}=u_i(1^k, (o_i=s,o_j=s))$ (where $i\neq j$) and so on. 

\begin{table}[htbp]
\caption{Outcomes and Utilities for $(2,2)$ rational secret reconstruction}%\footnote{The notations (e.g., $U_1$ , $U_1^-$ etc.) in brackets for the last two columns represent the corresponding notations used in [2] and [21].}%
\label{table: OutcomesRSS}
\begin{center}
\begin{tabular}{llll}
\hline\noalign{\smallskip}
%$P_1$'s outcome $(o_1)$ & $P_2$'s outcome $(o_2)$ & $P_1$'s Utility $U_1(o_1, o_2)$ & $P_2$'s Utility $U_2(o_1, o_2)$\\
$P_1$\rq s outcome & $P_2$\rq s outcome & $P_1$\rq s Utility & $P_2$\rq s Utility\\
$(o_1)$ & $(o_2)$ & $U_1(o_1, o_2)$ & $U_2(o_1, o_2)$\\
\hline
\noalign{\smallskip}
$o_1$=$s$ & $o_2$=$s$ & $U_1^{TT}$ & $U_2^{TT}$\\
$o_1$=$\perp$ & $o_2$=$\perp$ & $U_1^{NN}$ & $U_2^{NN}$\\
$o_1$=$s$ & $o_2$=$\perp$ & $U_1^{TN}$ & $U_2^{NT}$\\
$o_1$=$\perp$ & $o_2$=$s$ & $U_1^{NT}$ & $U_2^{TN}$\\
$o_1$=$\perp$ & $o_2\not\in {\{s, \perp\}}$ & $U_1^{NF}$ & $U_2^{FN} $\\
$o_1\not\in {\{s, \perp\}}$ & $o_2$=$\perp$ & $U_1^{FN} $ & $U_2^{NF}$\\
\hline
\end{tabular}
\end{center} 
\end{table}

For quantum domain, the secret is a state $\ket{\psi} = \alpha\ket{0} + \beta\ket{1}$, or in other words, a pair of complex numbers $(\alpha, \beta)$. 
Thus, the size of the secret is effectively infinite. Hence, the assumption
on the utilities as polynomial functions of the security parameter has no meaning. Rather, we treat the utilities as real numbers that depend on the output
values.

Players have their preferences based on the different possible outcomes. In this work, a rational player $i$ is assumed to have the following preference:
$$\mathcal{R}_1 : U_i^{TN}> U_i^{TT}>U_i^{NN}>U_i^{NT}.$$ 
Some players may have the additional preference $$U_i^{NF} \geq U_i^{TT},$$ whereas the rest have $$U_i^{NF} < U_i^{TT}.$$

For more than two players, the second superscript $Y$ in the notation 
$U_i^{XY}$ correspond to any of the other players (except $i$ itself).

\subsection{Fairness}
A rational player, being selfish, desires an unfair outcome, i.e., obtaining the secret alone. Therefore, the basic aim of rational secret sharing schemes has been to achieve fairness. A formal definition of fairness in the context of a
(2,2) RSS protocol was presented by Asharov and Lindell~\cite{AL}. We modify this definition for the $(t,n)$ quantum setting as follows:
\begin{definition}
(Fairness, adapted from \cite{AL}) A rational secret reconstruction mechanism $(\Gamma,\overrightarrow{\sigma})_{t,n}$ is said to be completely fair if for every arbitrary alternative strategy $\sigma'_i$ followed by party $P_i$, $(i\in{\{1,\ldots, n\}})$ the following holds:
$$
Pr[o_i (\Gamma,(\sigma'_i,\overrightarrow{\sigma}_{-i}))=s] < Pr[o_{-i} (\Gamma,(\sigma'_i,\overrightarrow{\sigma}_{-i}))=s].
$$
\end{definition}
In the above definition, the subscript $-i$ denotes all the players other than 
$i$.

Fairness can be achieved by a suitable randomized reconstruction of the protocol. The exact round in which the actual secret is to be revealed is not known to the parties. In Theorem~\ref{fair}, we show that the condition for fairness is
$$\gamma U_i^{TN}+(1-\gamma)U_i^{NN}<U_i^{TT}.$$
$$\mbox{Or, }\indent\gamma < \frac{U_i^{TT} - U_i^{NN}}{U_i^{TN} - U_i^{NN}}$$
for each $i$.
This is the same condition that is required in the classical scenario.

\subsection{Correctness}\label{correctness}
A formal definition of correctness in the context of a
(2,2) RSS protocol was presented by Asharov and Lindell~\cite{AL}. We modify this definition for the $(t,n)$ quantum setting as follows:
\begin{definition}
(Correctness) A rational secret reconstruction mechanism $(\Gamma,\overrightarrow{\sigma})$ is said to be correct if for every arbitrary alternative strategy $\sigma'_i$ followed by party $P_i$, $(i\in {\{1,\ldots,n\}})$ the following holds:
$$
Pr⁡[o_{-i} (\Gamma,(\sigma'_i,\overrightarrow{\sigma}_{-i}))\not\in \left\{{s,\perp}\right\} ]= 0
$$
\end{definition}
In classical rational secret sharing, the condition of correctness becomes significant in the non-simultaneous channel model. The rational party with preference $\mathcal{R}_1$ communicating last in any round may quit early in the protocol. Since other parties decide whether the revelation round has been reached depending on whether the last party has quit, they are easily misled into believing in a wrong value of the secret. 

\subsection{Equilibrium}
A rational secret reconstruction protocol should be such that no player has any incentive to deviate from this protocol. Consequently, Nash equilibrium and its several variants have been used as the equilibrium concept in the literature of rational secret sharing. A suggested strategy $\overrightarrow{\sigma}$ of a mechanism $(\Gamma,\overrightarrow{\sigma})$ is said to be in Nash equilibrium when there is no incentive for a player $P_i$ to deviate from the suggested strategy, given that everyone else is following this strategy. 

The concept of strict Nash equilibrium becomes useful when the payoffs from playing a \lq good\rq\ strategy and a \lq bad\rq\ strategy are so close that any minor changes in the beliefs of players about the strategy others are going to adopt may lead each of them to play the \lq bad\rq\ strategy \cite{KN1}. It is defined as follows:
\begin{definition}
(Strict Nash equilibrium) The suggested strategy $\overrightarrow{\sigma}$ in the mechanism $(\Gamma,\overrightarrow{\sigma})$ is a strict Nash equilibrium if for every $P_i$ and for any strategy $\sigma'_i$, we have $u_i (\sigma'_i,\overrightarrow{\sigma}_{-i} )<u_i (\overrightarrow{\sigma})$.
\end{definition}
There may exist several strategies which are the same as the suggested strategy $\overrightarrow{\sigma}_i$ for party $P_i$ except for minor differences such as performing some irrelevant computation or sending different messages after the protocol is over. For the sake of proving that a proposed protocol is in strict Nash equilibrium, we assume that all such strategies are essentially the same and do not constitute any deviation. 
%It is easy to see that the condition for fairness, i.e., $$\gamma U_i^{TN}+(1-\gamma)U_i^{NN}<U_i^{TT}$$ is sufficient to achieve strict Nash equilibrium in classical rational secret sharing schemes.

\section{Quantum Rational Secret Sharing}
In this section we first present a $(3,7)$ quantum rational secret sharing (QRSS) protocol and we generalize it to the $(t,n)$ setting in the next section.

We do not exploit the ideas related to teleportation in
quantum secret sharing. The idea of teleportation does not naturally take
care of the situation when parties are rational. Rather, we use quantum error
correcting code. There exist some works~\cite{C05,rietjens05,Z06,S09} for building quantum secret sharing schemes using (classical or quantum) error correcting codes. However, none of these schemes addresses the rationality issue.  

An arbitrary pure single-qubit quantum state is given by 
$\ket{\psi} = \alpha\ket{0} + \beta\ket{1}$ with 
$|\alpha|^2 + |\beta|^2 = 1$, where $\alpha, \beta \in \mathbb{C}$. 
Quantum error correction scheme known as CSS code~\cite{nc04},
can be constructed from classical error correcting code. Let $C$ and $C_1$ be two classical linear codes such that $\{0\}\subset {C}_{1} \subset C \subset \mathbb {F}_{2}^{7}$ with the generator matrices 

$$ G =
\begin{pmatrix}
1&0&0&0&0&1&1\\
0&1&0&0&1&0&1\\
0&0&1&0&1&1&0\\
0&0&0&1&1&1&1
\end{pmatrix},
$$

$$ G_1 =
\begin{pmatrix}
0&0&0&1&1&1&1\\
0&1&1&0&0&1&1\\
1&0&1&0&1&0&1
\end{pmatrix}.
$$

From the above expression it is clear that  $C_1$  is the dual 
code~\cite{mac} of $C$. 
A quantum CSS code can be constructed from these two linear codes with code 
words $\ket{0}_L$ and $\ket {1}_L$. A pure single quantum state can be encoded with this code by attaching an ancilla state $\ket{0}$ and applying the CNOT gate. After inserting the ancilla state we get $\alpha\ket{00} + \beta \ket{10}$ which is converted to $\alpha\ket{00} + \beta\ket{11}$ after the application of the CNOT gate. $\ket {00}$ can be encoded by the above CSS code as $\ket{1}_L$  and $\ket{11}$ can be encoded as $\ket{0}_L$. Thus the entire state is encoded as
\begin{eqnarray*}
\frac{1}{\sqrt{8}}[[\alpha\ket{1111111}  + \ket{1010010} + \ket {1100100} +
\ket{1001001}\\
+  \ket{0000111}+ \ket{0101010} + \ket{0011100} + \ket{0110001}]\\
  +  \beta[\ket{0000000} + \ket{0101101} + \ket{0011011} + \ket{0110110} \\
+ \ket{1111000} + \ket{1010101} + \ket{1100011} + \ket{1001110}]].
\end{eqnarray*}

In light of the above discussion, let us now explain our exact proposal
and its importance. Here, we assume the secret as a single qubit $\alpha\ket{0} + \beta\ket{1}$. We encode the secret with the above CSS code. Thus the secret is now split into seven qubits. The dealer distributes these seven qubits among seven parties. We now write the secret as 
\begin{eqnarray*}
\frac{1}{\sqrt{8}}[\ket{1111}[\alpha\ket{111} +  \beta\ket{000}] 
+ \ket{1010}[\alpha\ket{010} + \beta\ket{101}] \\
+ \ket {1100}[\alpha\ket{100} + \beta\ket{011}] 
+  \ket{1001}[\alpha\ket{001} + \beta\ket{110}] \\
+  \ket{0000}[\alpha\ket{111} + \beta\ket{000}] 
+ \ket{0101}[\alpha\ket{010} + \beta\ket{101}] \\
+ \ket{0011}[\alpha\ket{100} + \beta\ket{011}] 
+ \ket{0110}[\alpha\ket{001}] + \beta\ket{110}]]. 
\end{eqnarray*}
Applying CNOT gate on last three qubits, we obtain 
\begin{eqnarray*}
\frac{1}{\sqrt{8}}[\ket{1111}[\alpha\ket{1} +   \beta\ket{0}]\ket{00}  + \ket{1010}[\alpha\ket{0} + \beta\ket{1}]\ket{10} \\
+ \ket {1100}[\alpha\ket{1} + \beta\ket{0}]\ket{11} +  \ket{1001}[\alpha\ket{0} + \beta\ket{1}]\ket{01} \\
+  \ket{0000}[\alpha\ket{1} + \beta\ket{0}]\ket{00} + \ket{0101}[\alpha\ket{0} + \beta\ket{1}]\ket{10} \\
+ \ket{0011}[\alpha\ket{1} + \beta\ket{0}]\ket{11} + \ket{0110}[\alpha\ket{0}] + \beta\ket{1}]\ket{01}].
\end{eqnarray*}

Thus if last three parties collaborate, then one can reconstruct the secret by
measuring the last two qubits in $\{00,01,10,11\}$  basis. If a party gets $\ket{00}$ or $\ket{11}$, he has to apply $X$ gate to obtain the secret. If a party gets $\ket{01}$ or $\ket{10}$, he has to apply $I$ gate to construct the secret. Close observation reveals that there are only seven combinations of three parties for which secret can be reconstructed. Denoting the position of the participants by integer values, we can write those combinations as  $\mathcal{A} = $$\{(5,6,7)$, $(1,2,5)$, $(2,4,6)$, $(1,3,6)$, $(1,4,7)$, $(2,3,7)$, $(3,4,7)\}$. The set $\mathcal{A}$ is called the access structure of the secret sharing scheme.

\subsection{(3,7) Quantum Rational Secret Sharing Protocol}
In classical rational secret sharing, an indicator is distributed along the shares of each secret to each party. The parties reconstruct the indicator and comes to know about the revelation round. However, in quantum domain, including an indicator is costly. We solve this problem by 
assuming that the dealer is semi-offline. In other words, the dealer interacts with the participants twice, 1) at the time of the share distribution, and 2) at the time when the game is over. In Section~\ref{offline}, we discuss how to make the dealer offline. Like the classical RSS protocols, our dealer is assumed to be honest. 

In classical simultaneous broadcasting channel, each party is supposed to broadcast his share in each round. So, in each round, each party obtains $(t-1)$ shares from others and thus reconstructs the secret. In the point-to-point channel model, instead of broadcasting his share, each party in each round individually communicates his share to every other party. This means that each party prepares $t$ copies of his or her share and distributes $(t-1)$ shares among $(t-1)$ parties retaining one share for himself. %In this model also each party gets $t$ number of shares and can reconstruct the secret as the generation of the secret to every player is the fundamental requirement for the fairness. 
In quantum domain, due to the no cloning theorem~\cite{noclone}, a player cannot generate copies of his share. However, the dealer can prepare as many copies of the secret as required as he knows the secret. We exploit this idea to form our protocol. The communication in the quantum setting is similar to the communication in the point-to-point channel model. Unlike classical RSS,  each round is further sub divided into sub-rounds. In the $i$th sub-round of a round $j$ the participant $P_i$ is given the current shares (qubits) by the remaining players. For example, in $(3,7)$ quantum rational secret sharing, in the first sub-round $P_2$ and $P_3$ give their current shares (qubits) to $P_1$. In the second sub-round $P_1$ and $P_3$ give their current shares to $P_2$. In the third sub-round $P_1$ and $P_2$ give their current shares to $P_3$.

We assume that the players are of fail-stop nature. This means that they do not send wrong shares. In each round, a player has just two strategies, either to send his share or to remain silent. Remaining silent is equivalent to quit the game. Throughout the paper whereas we use the word ``quit", we want to mean that the player remains silent from the very sub-round of a round. The dealer is assumed to be honest.
We describe our protocol $\pi^{3,7}_{QRSS}$ below. Without loss of generality,
we assume that the dealer wishes to reveal the secret to the parties
receiving qubits 5, 6 and 7 and we label them as players $P_1$, $P_2$ and $P_3$
respectively.\\ \\
\textbf{\underline{[1. Protocol $\pi^{3,7}_{QRSS}$]}}\\
\textbf{1.1 The Dealer\rq s Protocol $\pi^{3,7}_{ShareGen}$:}\\
\textbf{Input.} The quantum secret to be shared using $(3,7)$ threshold secret sharing.\\
\textbf{\textit{1.1.1 Share Generation:}} The dealer does the following:
\begin{itemize}
\item Chooses the $r$ according to a geometric distribution $\mathcal{G}(\gamma)$ with parameter $\gamma$.  %For each round dealer tosses a biased coin. If head comes, he distributes the shares of the actual secret. Otherwise he distributes the shares of the fake secret. The round in which dealer distributes the shares of the actual secret is called revelation round. Let the probability that head comes up be $\gamma$. Thus the probability that a round would be revelation round is $\gamma$. Dealer sets $\gamma$ according to the utility values of the player.
\item Generates three copies of each secret (fake as well as actual) for each round and encodes those in CSS code discussed above.
\item Prepares a list $list_i$ of shares for each party $P_i$ such that:
\begin{itemize}
\item Qubits $(5,6,7)$ of each secret are given to players $P_1$, $P_2$, and $P_3$ respectively so that each party gets three qubits, i.e. $P_1$ possesses three of $5$, $P_2$ has three of $6$ and $P_3$ possesses three of $7$ for each round.
\item Each list contains $3(r+w)$ shares, where $w$ is also chosen according to $\mathcal{G}(\gamma)$. % as each round is subdivided into $3$ sub rounds.%Each player is instructed to give his current share to others.
\end{itemize}
\end{itemize}
\textbf{Output.} The dealer distributes $list_i$ to party $P_i$.\\ \\
\textbf{\textit{1.1.2 Unmasking of Revelation Round $\pi^{3,7}_{Unmask}$:}}\\
\textbf{Input.} Signal $sig_i$ from each player $P_i$, $(i\in{1,2,3})$. % indicating the end of $\pi^{3,7}_{Recon}$.\\
\textbf{Computation and Communication.} The dealer does the following:
\begin{itemize}
\item If $sig_1=sig_2=sig_3=1$, announces the value of $r$ to each $P_i$.
\item If for at least one value of $i$, $sig_i\neq1$, aborts after announcing \textit{abort} to each party $P_i$.
\end{itemize}
\textbf{Output.} The dealer outputs either $r$ or \textit{abort} depending on the values of $sig_i$.\\ \\
\textbf{1.2 The Player\rq s Protocol}\\
\textbf{Input.} The list of shares $list_i$.\\
\textit{\textbf{1.2.1 Secret Reconstruction $\pi^{3,7}_{Recon}$:}}\\
\textbf{Computation and Communication.}
Each player $P_i$ does the following:
\begin{itemize}
\item In the $i$th sub-round of a round $j$, $P_i$ is given the current shares (qubits) by other two parties. 
\item Checks to see if the number of shares received is less than two. If yes, aborts and sends $sig_i=0$ to dealer. Else, continues.
\item At the end of round $j$, does the following:
\begin{itemize}
\item Applies CNOT gate considering his current share (qubit) as control. 
\item Measures target bits in $\{00, 01,10, 11\}$ basis. 
\item Depending on the measured value operates either $X$ gate or $I$ gate. 
\end{itemize}
\item Stores the secret $s_j$ obtained in the $j$th round. 
\item If $j= {\frac{1}{3}}|list_i|$, sends $sig_i=1$ to the dealer. If the dealer sends \textit{abort}, then aborts; else if the dealer sends the value of $r$ stores only $s_r$ and quits. If $j<{\frac{1}{3}}|list_i|$, continues.
%\item There are $r$ number of rounds.
\end{itemize}
\textbf{Output.} The quantum secret $s_r$.

\section{Generalization to ($t,n$) Quantum Rational Secret Sharing}
Before going to the ($t,n$) quantum rational secret sharing, first we show an existential result.

\begin{theorem}
 Let  $C = [n,k,d]$ and $C_{1} = [n,k-1,d^{\prime}]$ be two linear codes such that $C_{1} \subset {C}$. Given  $C\cdot C_{1}^{T} = 0$, there is always a secret sharing scheme provided $k$ is a power of 2.
\end{theorem}
\begin{proof}
From the definition it is clear that we always construct a CSS code from given $C$ and $C_1$ with codewords $\ket{0}_L$ and $\ket{1}_L$. Let the secret be $\alpha \ket{0} + \beta \ket{1}$. Let us take the tensor product of the secret and $(m-1)$ number of the ancilla states, where $m$ is any integer value. The final state becomes $\alpha \ket{0_{1} 0_{2} \dots 0_{m}} + \beta \ket{1_{1} 0_{2} \dots 0_{m}}$. 

Applying CNOT gate we obtain $\alpha \ket{0_{1} 0_{2} \dots 0_{m}} + \beta \ket{1_{1} 1_{2} \dots 1_{m}}$. $ \ket{0_{1} 0_{2} \dots 0_{m}}$ can be written in matrix form by a $2^{m}$ binary vector, $ (1 0 0 \dots 0 0)$. Similarly $\ket{1_{1} 1_{2} \dots 1_{m}}$ can be represented by a $2^{m}$ binary vector, $(0 0 0 \dots 0 1)$. In secret sharing scheme $ \ket{0_{1} 0_{2} \dots 0_{m}}$ and $ \ket{1_{1} 1_{2} \dots 1_{m}}$ are message states. Thus $k = 2^{m}$ or $m = \log_{2} k$. As $m$ is an integer, $k$ should be power of 2.    
 The encoded secret in CSS code becomes $\alpha \ket{1}_L + \beta \ket{0}_L$.
\end{proof}

The next result gives a bound on the number of parties that can reconstruct 
the secret.
\begin{theorem}\label{vmezofelint}
 Let  $C = [n,k,d]$ and $C_{1} = [n,k-1, d^{\prime}]$ be two linear codes such that $C_{1} \subset {C}$. Given  $C\cdot C_{1}^{T} = 0$, there  are minimum $d$ number of parties who can reconstruct the secret.
\end{theorem}
\begin{proof}
According to the definition of the CSS code, $\ket{0}_L = \frac{1}{\| {C_1}\|} \sum_{x \in {C_1}} {\ket{x+C_1}}$ and  $\ket{1}_L = \frac{1}{\| {C_1}\|} \sum_{x \notin {C_1}} {\ket{x+C_1}}$. Thus, the codewords which consist the codeword $\ket{0}_L$ belong to $C_1$. On the other hand, the codewords consisting the codeword $\ket{1}_L$ belong to $C/C_1$. So we always get two orthogonal codewords which come from two different cosets. One codeword is associated with $\alpha$ and other codeword is associated with $\beta$. Let they be $u$ and $v$. Since $dist(u-v) \geq {d}$, if we apply CNOT gate on these $d$ bits considering the first bit as control, we get $(d-1)$ target bits which are equal in both the codewords. 
So measuring those $(d-1)$ qubits in $\{ 0,1\} ^{(d-1)}$ basis we can reconstruct the secret depending on the measurement result. Thus, for secret reconstruction minimum $d$ parties have to collaborate. 
\end{proof}

Note that since $C_1 \subset C$, we must have $d' \geq d$. For $(t,n)$
QRSS, we set $t = d$.

\subsection{$t$-out-of-$n$ Quantum Rational Secret Sharing Protocol}
Let us now present our generalized scheme.\\

\noindent\textbf{\underline{[2. Protocol $\pi^{t,n}_{QRSS}$]}}\\
\textbf{2.1 The Dealer\rq s Protocol $\pi^{t,n}_{ShareGen}$:}\\
\textbf{Input.} The quantum secret to be shared using $(t,n)$ threshold secret sharing.\\
\textbf{\textit{2.1.1 Share Generation:}}
The dealer does the following:
\begin{itemize}
\item The dealer designates $t$ players among $n$, from the access structure.
\item Chooses the $r$ according to a geometric distribution $\mathcal{G}(\gamma)$ with parameter $\gamma$.  %For each round dealer tosses a biased coin. If head comes, he distributes the shares of the actual secret. Otherwise he distributes the shares of the fake secret. The round in which dealer distributes the shares of the actual secret is called revelation round. Let the probability that head comes up be $\gamma$. Thus the probability that a round would be revelation round is $\gamma$. Dealer sets $\gamma$ according to the utility values of the player.
\item Generates $t$ copies of each secret (fake as well as actual) for each round and encodes those in CSS code derived from $C$ and $C_1$ (see \autoref{vmezofelint}).
\item Prepares a list $list_i$ of shares for each party $P_i$ such that:
\begin{itemize}
\item Each player $P_i$ is given a qubit from a valid set of $t$ qubits from the access structure like $(3,7)$ QRSS.
\item Each list contains $t(r+w)$ shares, where $w$ is also chosen according to $\mathcal{G}(\gamma)$. % as each round is subdivided into $3$ sub rounds.%Each player is instructed to give his current share to others.
\end{itemize}
\end{itemize}
\textbf{Output.} The dealer distributes $list_i$ to party $P_i$.\\ \\
\textbf{\textit{2.1.2 Unmasking of Revelation Round $\pi^{t,n}_{Unmask}$:}}\\
\textbf{Input.} Signal $sig_i$ from each player $P_i$, $(i\in{1,\ldots,t})$.\\% indicating the end of $\pi^{t,n}_{Recon}$.\\
\textbf{Computation and Communication.} The dealer does the following:
\begin{itemize}
\item If $sig_i=1$ for all $i \in {1,\ldots,t}$, announces the value of $r$ to each $P_i$.
\item If for at least one value of $i$, $sig_i\neq1$, aborts after announcing \textit{abort} to each party $P_i$.
\end{itemize}
\textbf{Output.} The dealer outputs either $r$ or \textit{abort} depending on the values of $sig_i$.\\ \\
\textbf{2.2 The Player\rq s Protocol}\\
\textbf{Input.} The list of shares $list_i$.\\
\textit{\textbf{2.2.1 Secret Reconstruction $\pi^{t,n}_{Recon}$:}}\\
Each player $P_i$ does the following:
\begin{itemize}
\item In the $i$th sub-round of a round $j$, $P_i$ is given the current shares (qubits) by other $(t-1)$ players. 
\item Checks to see if the number of shares received is less than $(t-1)$. If yes, aborts and sends $sig_i=0$ to dealer. Else, continues.
\item At the end of round $j$, does the following:
\begin{itemize}
\item Applies CNOT gate considering his current share (qubit) as control. 
\item Measures target bits in $\{0,1\}^{(t-1)}$ basis. 
\item Depending on the measured value operates either $X$ gate or $I$ gate. 
\end{itemize}
\item Stores the secret $s_j$ obtained in the $j$th round. 
\item If $j={\frac{1}{t}|list_i|}$, sends $sig_i=1$ to the dealer. If the dealer sends \textit{abort}, then aborts; else if the dealer sends the value of $r$ stores only $s_r$ and quits. If $j<{\frac{1}{t}|list_i|}$, continues.
%\item There are $r$ number of rounds.
\end{itemize}
\textbf{Output.} The quantum secret $s_r$.

In the next result, we show that the fairness condition for classical domain
remains valid in the quantum domain as well.
\begin{theorem}
\label{fair}
If $\gamma>0$ and $U_i^{TT}> \gamma U_i^{TN} + (1-\gamma) U_i^{NN}$, 
the protocol $\pi^{t,n}_{QRSS}$ achieves fairness.
\end{theorem}
\begin{proof}
A player who wants to obtain the secret alone must be able to correctly guess which round is the revelation round. Suppose the $i$th player guesses that the $j$th round is the revelation round and quits in the $i$th sub-round of the $j$th round. On other words, the player remains silent from the $(i+1)$th sub-round of $j$th round. Our protocol is designed in such a way that if any player quits in any intermediate round, then it is reported to the dealer by a signal bit ($sig$). If for at least one value of $i$, $sig_i\neq1$, the dealer aborts after announcing \textit{abort} to each party $P_i$.Thus, if the guess of the $i$th player is correct i.e $j = r$,the probability of which is $\gamma$, his utility is $U_i^{TN}$, else his utility is $U_i^{NN}$. So the expected utility of the player who decides to deviate based on his guess is given by $\gamma U_i^{TN}+(1-\gamma)U_i^{NN}$. On the other hand, if he simply followed the protocol, his utility would have been $U_i^{TT}$. However, the dealer chooses the value  $\gamma$ in such a way so that
$$\gamma U_i^{TN}+(1-\gamma)U_i^{NN}<U_i^{TT}.$$ Thus the player should have no incentive to deviate. He always gives his share and hence the protocol achieves fairness.
\end{proof}

The next result establishes the correctness of our general scheme.
\begin{theorem}
Even if some players may have $U_i^{NF} \geq U_i^{TT}$,
the protocol $\pi^{t,n}_{QRSS}$ achieves correctness.
\end{theorem}
\begin{proof}
In our protocol, the dealer is semi-offline. The revelation round is unmasked by him after all the players report that the reconstruction game is over. Therefore, players do not depend on the action of the last player to know which round is the revelation round (see \autoref{correctness}). Thus, the protocol is $U_i^{NF}$ independent. No  player has any incentive to quit in an intermediate round in the purpose to make the others believe in a fake secret as actual secret. 
Hence the protocol is correct.
\end{proof}
Note that in~\cite{AL}, it is mentioned that for non-simultaneous channel model,
$U_i^{NF}$ independence is impossible. But there the underlying assumption is
that the dealer is offline. Since in our QRSS, the dealer is semi-offline, we 
can easily get $U_i^{NF}$ independence even in non-simultaneous channel.

Now, we can state the following result on equilibrium.
\begin{theorem}
\label{nash1}
If $\gamma>0$ and $U_i^{TT}> \gamma U_i^{TN} + (1-\gamma) U_i^{NN}$, then protocol $\pi^{t,n}_{QRSS}$ achieves strict Nash equilibrium.
\end{theorem}
\begin{proof}
Let us assume that a party $P_i$ follows the deviating strategy $\sigma'_i$, when all other parties follow the protocol. Suppose $P_i$ aborts at round $j$ and the revelation round is $r$. Then either $j$ is itself the revelation round, i.e., $r=j$ or it is a round before the revelation round, i.e.  $r>j$. In our case, the secret reconstruction game is played until each party exhausts his list of shares. After that the dealer points out the revelation round. Hence, there is a possibility that $P_i$ deviates after the revelation round, i.e., $r<j$ although such deviation is not helpful in any way to the deviating party. 
We assume that the correct secret can be obtained by a player only when he quits in the revelation round. 
From the property of geometric distribution, we have

\qquad\qquad $\gamma = \Pr[j=r|j\leq r]=\frac{\Pr[j=r]}{\Pr[j \leq r]}$, and 

\qquad\qquad $1-\gamma = \Pr[j<r|j\leq r]=\frac{\Pr[j<r]}{\Pr[j \leq r]}.$

Then, $u_i(\sigma'_i,\overrightarrow{\sigma}_{-i})$ is given by 
\begin{eqnarray*}
&& U_i^{TN}\Pr[j=r] + U_i^{NN}\Pr[j<r] + U_i^{TT}\Pr[j>r]\\
& = & \gamma U_i^{TN}\Pr[j\leq r] + (1-\gamma)U_i^{NN}\Pr[j\leq r]\\
&& + U_i^{TT}(1-\Pr[j\leq r]) \\
& = & (\gamma U_i^{TN} + (1-\gamma)U_i^{NN}- U_i^{TT})\Pr[j\leq r] + U_i^{TT}\\ 
& < & U_i^{TT}.
\end{eqnarray*}
The last inequality follows from our assumption that 
$\gamma U_i^{TN} + (1-\gamma) U_i^{NN} < U_i^{TT}$, which makes the term added to $U_i^{TT}$ negative.
In each sub-round, a player can send only a unique share (namely, the correct
share) as we have $u_i(\sigma'_i,\overrightarrow{\sigma}_{-i})<u_i(\overrightarrow{\sigma})$. So, our protocol follows strict Nash equilibrium.
\end{proof}

In the classical rational secret sharing domain, Halpern and Teague~\cite{HT}
claimed that (2,2) rational secret sharing scheme cannot be constructed.
Later, Gordon and Katz~\cite{GK} showed a general $(t,n)$ rational 
secret sharing that works for $n = t = 2$, thus refuting the claim of~\cite{HT}.
Our ($t,n$) quantum rational secret sharing scheme also
works for $t = 2$. In that case, we require [$n,k,2$] linear code to 
construct the CSS code. In principle, ($2,n$) QSS is the quantum analogue of
($2,2$) classical secret sharing, since in the quantum domain the dealer
designates, from the access structures, a specific subset of 2 players 
(out of $n$) except whom no one else can obtain the secret.

\section{Offline Dealer for QRSS}
\label{offline}
In this section, we propose a $(t,n)$ quantum rational secret sharing scheme with the dealer offline, i.e., after distributing the shares, the dealer does
not come into the picture. 

Considering the dealer as semi-offline, we have shown that 
our protocol $\pi^{t,n}_{QRSS}$ becomes $U^{NF}$ independent and hence correct. 
In \cite{AL}, it is shown that in case of offline dealer and non-simultaneous channel model, the protocol becomes $U^{NF}$ dependent. Hence, 
achieving correctness is not guaranteed, when the preferences of the players 
are defined by $\mathcal{R}_1$. Thus,  in case of offline dealer, 
to achieve correctness (which we show later), we suitably redefine 
the preferences of the players as follows:
$$\mathcal{R}_2 : U_i^{TN}> U_i^{TT}>U_i^{NN}>U_i^{NT},$$ 
and  $$U_i^{NF} < U_i^{TT}$$
for all players $i$.
Note that we have to restrict the utilities so that no player can have $U_i^{NF} \geq U_i^{TT}$.

In our protocol $\pi^{t,n}_{QRSSDO}$ below, 
we use a Boolean indicator variable $b$ associated with each round.
The secret bit $b$ is distributed among the designated $t$ parties through a 
$(t,t)$ Shamir secret sharing to denote whether the previous round was 
revelation round ($b=1$) or not ($b=0)$. Later, we discuss how to move from
classical indicator to quantum indicator.
\\ \\
\textbf{\underline{[3. Protocol $\pi^{t,n}_{QRSSDO}$]}}\\
\textbf{3.1 The Dealer\rq s Protocol $\pi^{t,n}_{ShareGen}$:}\\
\textbf{Input.} The quantum secret to be shared using $(t,n)$ threshold secret sharing.\\
\textbf{\textit{3.1.1 Share Generation:}}
The dealer does the following:
\begin{itemize}
\item Sets $t=d$ and designates $t$ players among $n$.
\item Chooses the $r$ according to a geometric distribution $\mathcal{G}(\gamma)$ with parameter $\gamma$.  %For each round dealer tosses a biased coin. If head comes, he distributes the shares of the actual secret. Otherwise he distributes the shares of the fake secret. The round in which dealer distributes the shares of the actual secret is called revelation round. Let the probability that head comes up be $\gamma$. Thus the probability that a round would be revelation round is $\gamma$. Dealer sets $\gamma$ according to the utility values of the player.
\item Generates $t$ copies of each secret (fake as well as actual) for each round and encodes those in CSS code derived from $C$ and $C_1$ (see \autoref{vmezofelint}).
\item Prepares a list $list_i$ of shares for each party $P_i$ such that:
\begin{itemize}
\item Each element $e_{k,i}$ in the list $list_i$ consists of two parts: a qubit from a valid set of $t$ qubits from the access structure as in $(3,7)$ QRSS and a $(t,t)$ Shamir share of a Boolean value indicating whether the previous round was the revelation round.
\item Each list contains $k=t(r+w)$ shares, where $w$ is also chosen according to $\mathcal{G}(\gamma)$. % as each round is subdivided into $3$ sub rounds.%Each player is instructed to give his current share to others.
\end{itemize}
\end{itemize}
\textbf{3.2 The Player\rq s Protocol}\\
\textbf{Input.} The list of shares $list_i$.\\
\textit{\textbf{3.2.1 Secret Reconstruction $\pi^{t,n}_{Recon}$:}}\\
Each player $P_i$ does the following:
\begin{itemize}
\item In the $i$th sub-round of a round $j$, $P_i$ is given the current element (one qubit and one classical bit) in the lists by other $(t-1)$ players. 
\item If the number of elements received is less than $(t-1)$ or if a partial element has been received then aborts. Else, continues.
\item At the sub-round $i$ of round $j$, does the following:
\begin{itemize}
\item Stores the qubits obtained from the $(t-1)$ parties.
\item Reconstruct the Boolean value $b$ associated with that round. 
\item If $b = 1$, 
\begin{itemize}
\item Set $r = (j-1)$.
\item Applies CNOT gate considering his $(j-1)$th share (qubit) as control. 
\item Measures target bits in $\{0,1\}^{(t-1)}$ basis. 
\item Depending on the measured value operates either $X$ gate or $I$ gate. 
\item Stores the quantum secret $s_{r}$ obtained in the $(j-1)$th round. 
\end{itemize}
Else, continues.
\end{itemize}
%\item If $j=|list_i|$, sends $sig_i=1$ to the dealer. If the dealer sends \textit{abort}, then aborts; else if the dealer sends the value of $r$ stores only $s_r$ and quits. If $j<|list_i|$, continues.
%\item There are $r$ number of rounds.
\end{itemize}
\textbf{Output.} The quantum secret $s_r$.
\\
\\
We can make the above protocol fully quantum by replacing the indicator bits
0 and 1 by qubits $\ket{0}$ and $\ket{1}$ respectively. However, instead of 
Shamir's ($t,t$) share, $\ket{0}$ and $\ket{1}$ are encoded by CSS code just like the secret state $s$. The dealer distributes the list ($list_i$) to each player $P_i$, containing two qubits, first one for the secret and the second one for the round. In each sub-round $i$ of a round $j$, player $P_i$ reconstructs the qubit associated with that round. If it is $\ket{1}$, the player comes to know that the $(j-1)$th round was the revelation round. He reconstructs the secret for the $(j-1)$th round and discards other qubits obtained in the previous rounds.  

Note that unlike the protocol for semi-offline dealer, the game will be over when the first player gets $1$ (in classical) or $\ket{1}$ (in quantum). He has no incentive to send his shares in subsequent sub-rounds. He then just quits the game. The rest of the players then conclude that the revelation round has been occurred just before that round. Thus they also get the secret by operating CNOT gate and measuring the target bits for the last complete round.

Unlike semi-offline dealer, in this case the players need not to reconstruct the secret qubits for each round. Instead, they reconstruct the secret only for the revelation round. Moreover, the players are not forced to exhaust their lists of shares. The revelation round is the last round in this protocol as after the revelation round the players have no incentive to continue the game. The protocol is fair, correct and achieves strict Nash equilibrium.  

\begin{theorem}
If $\gamma>0$ and $U_i^{TT}> \gamma U_i^{TN} + (1-\gamma) U_i^{NN}$, 
the protocol $\pi^{t,n}_{QRSSDO}$ achieves fairness.
\end{theorem}
The proof is the same as \autoref{fair}.

\begin{theorem}
Provided $U_i^{TT}>  U_i^{NF}$, 
the protocol $\pi^{t,n}_{QRSSDO}$ achieves correctness.
\end{theorem}
\begin{proof}
As $U_i^{TT}>  U_i^{NF}$, no player has any incentive to mislead others to believe in a wrong secret as an actual secret when he himself does not get the real secret. Thus the protocol is correct.
\end{proof}

We show strict Nash equilibrium in the next result. 
\begin{theorem}
\label{nash2}
If $\gamma>0$ and $U_i^{TT}> \gamma U_i^{TN} + (1-\gamma) U_i^{NN}$, then protocol $\pi^{t,n}_{QRSSDO}$ achieves strict Nash equilibrium.
\end{theorem}
\begin{proof}
Let us assume that a party $P_i$ follows the deviating strategy $\sigma'_i$, when all other parties follow the protocol. Suppose $P_i$ aborts at round $j$ and the revelation round is $r$. Then either $j$ is itself the revelation round, i.e., $r=j$ or it is a round before the revelation round, i.e.,  $r>j$.
We assume that the correct secret can be obtained by a player only when he quits in the revelation round. 
From the property of geometric distribution, we have

\qquad\qquad $\gamma = \Pr[j=r|j\leq r]=\frac{\Pr[j=r]}{\Pr[j \leq r]}$, and

\qquad\qquad $1-\gamma = \Pr[j<r|j\leq r]=\frac{\Pr[j<r]}{\Pr[j \leq r]}.$

Then, $u_i(\sigma'_i,\overrightarrow{\sigma}_{-i})$ is given by 
\begin{eqnarray*}
&& U_i^{TN}\Pr[j=r] + U_i^{NN}\Pr[j<r]\\
& = & \gamma U_i^{TN}\Pr[j\leq r] + (1-\gamma)U_i^{NN}\Pr[j\leq r]\\
& = & (\gamma U_i^{TN} + (1-\gamma)U_i^{NN})\Pr[j\leq r]  \\ 
& < & U_i^{TT}.
\end{eqnarray*}
The last inequality follows from our assumption that 
$\gamma U_i^{TN} + (1-\gamma) U_i^{NN} < U_i^{TT}$.
In each sub-round, a player can send only a unique share (namely, the correct
share) as we have $u_i(\sigma'_i,\overrightarrow{\sigma}_{-i})<u_i(\overrightarrow{\sigma})$. So, our protocol follows strict Nash equilibrium.
\end{proof}
Note that the proof in this
case is a bit different from the proof in Theorem~\ref{nash1}, because the 
case $r < j$ does not occur here.

\section{Conclusion and Future Work}
Quantum secret sharing schemes either derived from QECC or teleportation do not succeed when we assume rational players. In this paper, 
for the first time we propose a new quantum rational secret sharing schemes
that is fair, correct and achieves strict Nash equilibrium. Under this
scheme, we propose two protocols, one with semi-offline dealer and the other with offline dealer. Semi-offline dealer appears twice 1) at the time of the share distribution, 2) at the end of the game. In the second protocol,
we make the dealer offline by giving the players auxiliary information related to the revelation round.

The only disadvantage of the protocols is that they require quantum memory. Also, quantum no cloning theorem~\cite{noclone} resists the players to copy their shares. So we let the dealer can prepare the copies of the secret as he knows the secret. Removing the requirement of the quantum memory is an interesting open problem. 

In the classical rational secret sharing schemes, there are other notions
of Nash equilibrium, such as computational Nash 
equilibrium under different adversarial models~\cite{FKN,Katz,OPRV,LS}.
Due to CSS code structure, no cloning and infinite range of the secret,
the scenario is completely
different in quantum setting and hence we consider only strict Nash
equilibrium. Extended analysis of our protocol or its variants for alternative
equilibrium models could be another potential future work.

\end{document}